\documentclass[preprint,12pt,authoryear]{elsarticle}

\usepackage{amsfonts, amssymb, amsmath, amscd}
\usepackage{amsthm}
\newtheorem{theorem}{Theorem}
\newtheorem{corollary}{Corollary}
\usepackage{algorithm}
\usepackage{algpseudocode}
\usepackage{caption}
\usepackage{subcaption}

\journal{Information Fusion}

\begin{document}

\begin{frontmatter}



\title{Approximate Computational Approaches for Bayesian Sensor Placement in High Dimensions}


\author[cse]{Xiao Lin}
\ead{lin65@email.sc.edu}

\author[cse]{Asif Chowdhury}
\ead{asifc@email.sc.edu}

\author[visit]{Xiaofan Wang}
\ead{wangxfok@xaut.edu.cn }

\author[cse]{Gabriel Terejanu\corref{gabriel}}
\ead{terejanu@cec.sc.edu}

\address[cse]{Department of Computer Science and Engineering, 
University of South Carolina}

\address[visit]{Department of Computer Science and Engineering, 
Xi'an University of Technology}

\cortext[gabriel]{Corresponding author: Gabriel Terejanu. Address: 315 Main St, Swearingen Bldg. 3A01L, Columbia, SC 29208, USA.}

\begin{abstract}
Since the cost of installing and maintaining sensors is usually high, sensor locations should always be strategically selected. For those aiming at inferring certain quantities of interest (QoI), it is desirable to explore the dependency between sensor measurements and QoI. One of the most popular metric for the dependency is mutual information which naturally measures how much information about one variable can be obtained given the other. However, computing mutual information is always challenging, and the result is unreliable in high dimension. In this paper, we propose an approach to find an approximate lower bound of mutual information and compute it in a lower dimension. Then, sensors are placed  where highest mutual information (lower bound) is achieved and QoI is inferred via Bayes rule given sensor measurements. In addition, Bayesian optimization is introduced to provide a continuous mutual information surface over the domain and thus reduce the number of evaluations. A chemical release accident is simulated where multiple sensors are placed to locate the source of the release. The result shows that the proposed approach is both effective and efficient in inferring the QoI.
\end{abstract}


\end{frontmatter}


\section{Introduction}
\label{introduction}
Sensor placement plays an important role in a range of engineering problems, such as grid coverage~\citep{dhillon2003}, tracking target~\citep{cheng2013information,Hanbiao2004}, and monitoring spatial phenomena~\citep{MadankanSS14}. These problems may range across various subjects, however, many of them are similar in essence. In this paper, we restrict ourselves to the problem where quantities of interest (QoIs) need to be inferred given sensor measurements. One example in this category is monitoring atmosphere dispersion event. In a chemical release accident, it is desirable to know the release parameters such as location, strength and time, so as to aid emergency response. These release parameters are QoIs and can be inferred from sensor measurements. Other QoIs include concentration of chemicals which can be inferred in unobserved regions.

These sensor placement problems always involve mathematical models which describe the spatial and/or temporal process in concern. Usually,  QoIs are unknown model parameters and/or state variables.  The inference involves running computer simulations and solving inverse and forward problems. Generally, there are two categories of approaches, optimization methods and Bayesian inference. The optimization methods provide a single point estimate of unknown parameters by minimizing the discrepancy between model predictions and observational data. While Bayesian inference formulate the problem in a Bayesian framework where quantities have probabilistic description and posterior distribution of unknown parameters can be obtained via Bayes rule. 
The advantage of Bayesian inference is the incorporation of prior information and the access to the full posterior distribution from which estimates with quantified uncertainties can be extracted. Usually in Bayesian inference, Monte Carlo sampling are employed and inference problems are solved by Markov Chain Monte Carlo, particle filter or Ensemble Kalman filter for computation preference. For a thorough review of both optimization methods and Bayesian inference, one can refer to~\citep{hutchinson2017}. In this paper, we focus on Bayesian inference. 

Once we have measurement data, QoI can be inferred. However, data collected from different locations may provide different amount of information towards QoI. Moreover, it is usually impractical to place sensors exhaustively due to high cost of installing and maintaining sensors. Thus sensors should be placed judiciously so as to maximize the information content. Sensor locations are commonly decided by running simulation and maximizing certain target function. There are various choices towards the target function. It can be Mean Square Error (MSE), mutual information, entropy or measurements about inverse moment matrix in A-, D-, E-optimal design. Among all these criteria, mutual information is most common used, since it is a natural measure of dependence between the two variables and others can be derived from it either directly or indirectly. Ertin (2003)~\citep{Ertin2003} discussed maximum mutual information approach for dynamic sensor query problems. In his paper, mutual information between sensor data and target state was maximized at each step to decide which sensor was queried for tracking the target. And it was shown that this maximum mutual information approach was equivalent to minimizing expected posterior uncertainty in target state. This holds in the special case when the conditional distribution of the observable given the state is independent of the state as in the fixed additive Gaussian noise case. Andreas Krause (2008)~\citep{Krause2008} discussed sensor placements for prediction problems where mutual information between the observed locations and unobserved locations is maximized so as to gain most information about observed locations. Xiaopei Wu (2012)~\citep{Xiaopei2012} tackled a similar problem in soil moisture. Instead of applying maximum mutual information strategy globally, locations were first clustered according to soil moisture content, and then maximum mutual information was used in each cluster to select sensor locations.

In this paper, we apply the same idea that sensors should be placed where the mutual information between sensor measurements and QoI is maximized. The derivation of this metric is straight forward, however, computing mutual information is always challenging. Popular methods such as kernel density estimation and k nearest neighbor can give good estimate in low dimension, but work poorly in high dimensions due to the scarcity of samples.  This issue is not directly addressed in the past literature. In this paper, we propose an approximate sensor placement approach. Instead of maximizing mutual information directly, we derive a lower bound of mutual information, which can be computed in a much lower dimension. Furthermore, we adopt the same strategy as in \cite{Weaver2016} by using Bayesian optimization~\citep{Jones1998} to facilitate maximization over a continuous domain. This is much more efficient than the typical way which discretizes the domain into fine grid and make selection from all the grid points.  Finally, a chemical release accident is simulated and the proposed approach shows promising results.

The rest of the paper is organized as follows. In section 2, uncertainty modeling and Bayesian inference is introduced. The proposed sensor placement strategy is detailed in section 3. In section 4, a chemical release accident is simulated and the proposed approach is applied to infer release parameters. Conclusion is given in Section 5.

\section{Bayesian Inference}
\label{Bayesian Inference}

In this section, we introduce Bayesian inference for solving inverse problem. In the Bayesian framework, uncertainties in state variables and parameters are usually described with probability distributions. Measurement data is used to update the knowledge of these quantities. It is desirable for the posterior distribution to have a small uncertainty and at the same time to capture the true value. The connection between these quantities and observation data is embedded in the mathematical model which describes the process in concern. Since the proposed approach in this paper is universal, an abstract model will be introduced first.

\subsection{Uncertainty Modeling}

To set notation, consider the following abstract model:
\begin{align}
&\mathcal{R}(u, \theta, x) = 0 \label{eq:model} \\
&y = \mathcal{Y}(u, \theta, x) \label{eq:model_obs} \\
&q = \mathcal{Q}(u,\theta, x)~. \label{eq:QoI}
\end{align}
Here, $\mathcal{R}$ is some operator, $u$ is the solution or the state variable and $\theta$ is a set of parameters, which usually have a physical interpretation. $x$ denotes the scenario which defines the problem being considered. In sensor placement, $x$ usually refers to sensor locations. $\mathcal{Y}$ is a map from the solution to the prediction quantity $y$ that can be compared with sensor measurements $D$.  In addition, $\mathcal{Q}$ defines QoI which is denoted by $q$.  In our problem, $q$ can be $\theta$ or $u$ or other quantities inferred from $\theta$ and $u$. However, no matter what $q$ is, model parameters and state variables need to be known first. Then other quantities can be obtained through Eq.~\eqref{eq:QoI}. Usually, $\theta$ and $u$ are unknown or partially unknown, and need to be inferred from data. Let $\tau$ denotes unknown parts of $\theta$ and $u$. In this paper, Bayesian inference is carried out to solve the problem. 

Because the observation data is noisy due to sensor imprecision, the measurement noise $\epsilon$ usually follows a known pdf $p(\epsilon)$ that is defined by the specifications of the sensors. This results in the following relation between the observable $d$ and model prediction $y$.
\begin{equation}
d = y + \epsilon ~.\label{eq:meas_noise}
\end{equation}

Finally, the relation between the observable $d$ and model parameters $\theta$ as well as state variable $u$ is given by combining Eq.~\eqref{eq:model_obs} and Eq.~\eqref{eq:meas_noise}. This measurement model, Eq.~\eqref{eq:meas_model} defines the likelihood function and Bayes rule can be used to update the knowledge of $\tau$.
\begin{equation}
d = \mathcal{Y}(u, \theta, x) + \epsilon~.\label{eq:meas_model}
\end{equation}

Since Bayes rule is used as the inference engine, then a prior probability distribution needs to be defined for $\tau$, that is, $\tau \sim p(\tau)$. Note that the additive errors introduced in the previous equations are not a requirement; multiplicative errors or embedded errors are possible as well. 

\subsection{Bayesian inference}
This paper employs probability to represent uncertainty and Bayesian inference to update the uncertainty of unknown quantity $\tau$ in light of observation data. In Bayesian inference, one seeks a complete probabilistic description of $\tau$ that make the model consistent with the observation data, $D$. The solution to this problem is the posterior probability density function of $\tau$ and it is defined by Bayes’ Theorem,
\begin{equation}
p(\tau|D) = \frac{p(D|\tau)p(\tau)}{p(D)}~. 
\label{eq:bayes_1}
\end{equation}

Here, $p(D|\tau)$ is the likelihood function and it measures the agreement between the model output and the data for given values of the input $\tau$. The denominator in Eq.~\eqref{eq:bayes_1} is called the marginal likelihood or evidence. Overall, this is just a normalization constant that ensures that the solution to the Bayes' inverse problem, $p(\tau|D)$ is indeed a proper pdf that integrates to one.

In a more general scenario, data is collected over a period of time, and the knowledge of $\tau$ is updated after each measurement. Let $D= \left\{d_{1}, d_{2},..., d_{M}\right\}$ denote the measurement data collected at M time points. The final posterior distribution $p(\tau|D_{M})$ can be obtained recursively as shown in Eq.~\ref{eq:recursiveupdate}
\begin{align}
p(\tau|D_{M})&= \frac{p(D_{M}|\tau)p(\tau)}{p(D_{M})} \nonumber\\
&= \frac{p(d_{M},D_{M-1}|\tau)p(\tau)}{p(d_{M},D_{M-1})} \nonumber\\
&= \frac{p(d_{M}|D_{M-1},\tau)p(D_{M-1}|\tau)p(\tau)}{p(d_{M}|D_{M-1})P(D_{M-1})} \nonumber\\
&= \frac{p(d_{M}|D_{M-1},\tau)p(\tau|D_{M-1})}{p(d_{M}|D_{M-1})}~.
\label{eq:recursiveupdate}
\end{align}

As we can see, in the inference above, the posterior distribution obtained at current time will be the prior for the next time point. Note, the true value of model parameters remains unchanged, while other quantities like states usually change with time.  Another thing worth mentioning is that during the observation time period, sensors can either remain fixed or be mobile. This raises two categories of sensor placement problems, static sensors and mobile sensors. If mobile sensors are used, sensor locations need to be selected at each time point given the current status of the system. In this case, the above observation data $d_{i}$ are collected from different locations. 

A vast number of computational approaches have been invented to solve Eq.~\ref{eq:bayes_1}. For linear models with Gaussian distribution and additive white noise, Kalman filter is the most accurate and efficient method to solve the inverse problem. However, this is not the case on most occasions, where models can be nonlinear and the distribution of states and/or parameters are non Gaussian.  Then we need to use numerical sampling techniques, also known as Monte Carlo. Most commonly used methods include particle filter and Markov Chain Monte Carlo (MCMC)~\citep{khaleghi2013multisensor}. In particle filter, samples' weights are updated by calculating likelihood at each sample. For better estimation, more advanced procedures might be added, such as regularization and progressive correction~\citep{oudjane:2000}. On the other hand, MCMC has become a main computational workhorse in scientific computing from a large class of distributions. The most basic form of MCMC is Metropolis-Hasting (MH) algorithm~\citep{Mira:2001,Beck:2002}, which generates a sequence of correlated samples that form a Markov chain. Improved versions of MCMC such as delayed rejection adaptive metropolis  (DRAM) and transitional MCMC (TMCMC)~\citep{Haario:2006,Beck:2002,Ching:2007} are also used. Another important method is ensemble Kalman filter (EnKF). EnKF was first introduced by Evensen \cite{Evensen2009} and has been widely used in various applications due to its simplicity in both theory and implementation. It originates from Kalman filter but uses Monte Carlo approach to represent probability distributions. In EnKF, samples are called ensemble members. Each ensemble member is updated through similar formula as in Kalman filter. EnKF propagates a relatively small ensemble of samples through the system non-linearities and moves them such that their mean and covariance approximate the first two moments of the posterior distribution.  

As we can see, once observation data are collected, various methods can be performed to implement Bayesian inference. However, sensors placed at different locations might provide different amount of information about $\tau$. More specifically, the posterior distribution $p(\tau|D)$ in Eq.~\ref{eq:bayes_1} is more likely to be different given different $D$. Some might have a larger uncertainty, while others have more accurate and confident estimation. Thus placing sensors judiciously is important and this will be detailed in the following section.

\section{Sensor Placement}
\label{SensorPlacement}

The placing of sensors to infer the QoIs is formulated in Bayesian framework where each QoI has a probabilistic description. Here, we use $q$ to denote all the QoI which might include model parameters, state variables or other quantities. Suppose we have prior information $p(\tau)$, when observation data $d$ is collected, posterior $p(\tau|d)$ can be estimated  via Bayes' rule which has been discussed in the previous section. Usually, the uncertainty of $\tau$ will be reduced after Bayesian inference. Meanwhile, any uncertainty of $\tau$ will be propagated to $q$ through Eq.~\eqref{eq:QoI} and it is desirable for $q$ to have small uncertainty. However, data collected at different locations will lead to different posterior distribution $p(q)$ which sometimes can be too wide to provide useful information. Thus sensors should be placed strategically. Here, we develop a approach which selects the most informative sensor locations. The approach is based on mutual information criterion which will be discussed first.

\subsection{Mutual Information Criterion}

Mutual information criterion was first introduced by Lindley~\citep{Lindley1956} who used this criterion to measure the expected amount of information provided about the unknown parameters $\theta$ by the measurement data $d$ in an experiment. In this criterion, the information obtained from an experiment is quantified by the reduction in uncertainty of $\theta$ which is represented by Shannon entropy. Thus different experimental conditions can be compared and the one that leads to the most reduction in uncertainty of $\theta$ will be selected to perform the experiment. 

Sensor placement shares the same idea. We use $x$ to denote sensor locations, the amount of information provided by observation data $d$ at  $x$ will be  
\begin{align} \label{utilFunc}
U(d,x) =&- \int_{Q} p(q) \log p(q)
    \mathrm{d}q - (-\int_{Q} p(q|d)
    \log p(q|d) \mathrm{d}q)
\end{align}
where the first term on the right is the entropy of prior distribution $p(q)$ and the second term is the entropy of posterior distribution $p(q|d)$. Since observation data can only be obtained after sensors are placed, the expected amount of information is  calculated by marginalizing over all possible observations:
\begin{equation} \label{expUtil}
\mathrm{E}_{d}[U(d,x)] = \int_{\mathcal{D}} U(d,x)
    p(d,x) \mathrm{d} d~.
\end{equation}
Eq.~\eqref{expUtil} can be expanded as follows:
\begin{align}
\mathrm{E}_{d}[U(d,x)] =& \int_{\mathcal{D}}
    \int_{Q} p(q,d|x)
    \log \frac{p(q,d|x)}
    {p(d|x)}
    \mathrm{d}q \mathrm{d}d \nonumber \\
 & - \int_{\mathcal{D}} \int_{Q} p(q,d|x)
    \log p(q) \mathrm{d}q \mathrm{d}d \\
=&\int_{\mathcal{D}} \int_{Q} p(q,d|x)
    \log \frac{p(q,d|x)}
    {p(q)p(d|x)}
    \mathrm{d}q\mathrm{d}d \nonumber \\
=& I(q;d|x)~.
\end{align}
We can see that the expected amount of information of $q$ provided by sensors equals mutual information between $q$ and sensor readings $d$.  Therefore sensors should be placed where this mutual information is maximized. This is shown in Eq.\eqref{MI criterion} 
\begin{align}\label{MI criterion}
x^* = \arg\max_{x}
    I(d;q|x)
\end{align}
Although mutual information is a perfect criterion theoretically, estimating mutual information is challenging. Some commonly used estimators include histogram based estimator, kernel density estimator, and $k$-nearest neighbor estimator (kNN). In their survey, Walters-Williams and Li~\citep{WaltersWilliams:2009gh} show that parametric estimation usually outperform non-parametric estimation when data is drawn from a known family of distributions. But this is not the case in most practical problems. Usually mutual information is estimated directly from Monte Carlo samples. Khan et al.~\citep{Khan:2007up} compare different estimators and show that the kNN estimator of mutual information captures better the nonlinear dependence than other commonly used estimators. In our paper, we adopt the kNN estimator of mutual information proposed by Kraskov et al.~\citep{Kraskov:2004gr}, which is based on kNN estimator of entropy proposed by Kozachenko and Leonenko~\citep{Kozachenko:1987ts}.  The kNN estimator is as follows:
\begin{align}
I(q;d|x) \approx & -\frac{1}{N}\sum^{N}_{i=1}\left(\psi(n_{q}(i) + 1) + \psi(n_{d}(i) + 1)\right) + \psi(k) + \psi(N)~.
\end{align}
Here, $n_{d}(i)$ and $n_{q}(i)$ are the number of samples in the marginal space within the distance from the $i$th sample to its kNN in the joint space, and $\psi(k)$ is the digamma function which satisfy the recursion $\psi(k+1) = \psi(k) + 1/x$ and $\psi(1) = -C$, where $C \approx 0.5772156$ is the Euler-Masheroni constant. Note that a small value for $k$ will result in a small bias but a large variance and vice-versa. Also, the efficiency of the estimator decreases as the dimensionality of the joint space increases.  

\subsection{Sensor Placement Approach}

The main focus of this paper is to solve the problem of computing mutual information in high dimensions. Due to scarcity of samples used to capture the joint distribution $p(d,q)$ in high dimensions the numerical approximation of the mutual information $I(d;q)$ is not reliable. To solve this problem, in this section we develop an novel approach which computes the lower bound of $I(d;q)$ in much lower dimension. Also, in order to lower the computational cost, Bayesian optimization~\citep{Jones1998,Weaver2016} is introduced to generate a mutual information surface which greatly reduces the number of evaluations.

\subsubsection{Data Processing Inequality}

The derivation of lower bound is based on data processing inequality which is stated below~\citep{Cover:2006}.
\begin{theorem}\label{Data-processing inequality}
If random variables $J$, $V$, $Z$ form a Markov chain in the order denoted by $J \rightarrow V \rightarrow Z$, then $I(J; V) \geq I(J; Z)$.
\end{theorem}
$J$, $V$, $Z$ form a Markov chain means the conditional distribution of $Z$ depends only on $V$ and is conditionally independent of $J$, that is $p(Z|V) = p(Z|J, V)$. A special case of Theorem~\ref{Data-processing inequality} is that if $Z = g(V)$, then $I(J; V) \geq I(J; Z)$ which means no function $g(\bullet)$ can increase the amount of information that $V$ tells about $J$. We extend this idea by applying transformations on both $J$ and $V$, which leads to Corollary~\ref{corollary of DPI}.
\begin{corollary}\label{corollary of DPI}
For random variables $J$, $V$, $Z$ and $U$,  if $J \rightarrow V \rightarrow Z$ and $U = h(J)$, then $I(J; V) \geq I(U; Z)$ .
\end{corollary}
\begin{proof}
$J \rightarrow V \rightarrow Z$ implies $Z\rightarrow V \rightarrow J$ which further implies $Z\rightarrow J \rightarrow h(J)$. Thus $I(J; V) \geq I(J; Z) \geq I(U; Z)$. In particular, if $Z = g(V)$, then $I(J; V) \geq I(h(J), g(V))$.
\end{proof}
In our case, for any transformation $g(\bullet)$ on $d$, and  $h(\bullet)$ on $q$ we have $h(q) \leftrightarrow q \leftrightarrow d \leftrightarrow g(d)$ which leads to 
\begin{equation}\label{inequality}
I(q; d) \geq I(q; g(d)) \geq I(h(q); g(d))~.  \\
\end{equation}
Thus, the problem is converted to looking for two proper transformations $g(\bullet)$ on $d$ and $h(\bullet)$ on $q$. Since we want mutual information to be computed in lower dimension,  the most favorable transformations are the ones that convert $q$ and  $d$ to two one dimensional variables, at the same time keep the dependency between $q$ and $d$. One such transformation is canonical correlation analysis (CCA) which will be discussed in the following section.

\subsubsection{Canonical Correlation Analysis}

Given two random vectors $J$ and $V$, canonical correlation analysis is concerned with finding two directions on which the projections of $J$ and $V$ have maximum correlation.  Assuming $J$ has dimension $p_{1}$ and $V$ has dimension $p_{2}$, there are $p$ canonical correlations of $J$ and $V$, where $p = \min\lbrace p_{1}, p_{2} \rbrace$.  These $p$ canonical correlations can be obtained by solving a generalized eigenvalue problem~\citep{Bach:2003}. The solution will provide all $p$ correlations which are ranked in decreasing order. Here, we only consider the first canonical correlation, that is the largest one among all the correlations, and use it to derive the lower bound of mutual information. 

Let $\alpha$ and $\beta$ denote two projecting directions, the first canonical correlation $\rho_{1}$ is defined as the maximum possible correlation between two projections $\alpha^{T}J$ and $\beta^{T}V$: 
\begin{equation}
\rho_{1} = \max_{\alpha, \beta} corr(\alpha^{T}J, \beta^{T}V)~.
\end{equation}
Here $\alpha^{T}J$ and $\beta^{T}V$ are both one dimensional random variables.  Recall that our goal is to find two transformations $h(\bullet)$ and $g(\bullet)$ such that $I(h(J); g(V))$ can be computed in lower dimension. Thus we can choose $h(J) = \alpha^{T}J$ and $g(V) = \beta^{T}V$ where $\alpha$ and $\beta$ are projecting directions of the first canonical correlation. 

~\citep{Bach:2003} discussed relation between CCA and mutual information. It is shown that if the joint distribution of $J$ and $V$ is Gaussian, the mutual information between $J$ and $V$ can be expressed as a function of canonical correlations $\rho_{i}$, $i = 1, 2, ... , p$.
\begin{equation}
I(J, V) = -\dfrac{1}{2}\sum_{i=1}^{p}log(1 - \rho_{i}^{2})~.
\end{equation}
We can further denote $I_{i}(J, V)$ as the mutual information between the $i$th canonical projections. Then 
\begin{equation}
I(J, V) = \sum_{i=1}^{p}I_{i}(J, V)
\end{equation}
which means if $p = 1$, the mutual information of Gaussian variable $J$ and $V$ is equivalent to the mutual information of their first canonical projections, that is
\begin{equation}
I(J, V) = I_{1}(J, V) = I(\alpha^{T}J; \beta^{T}V )~.
\end{equation}
In our case,  $h(\bullet)$ and $g(\bullet)$ in Eq.~\eqref{inequality} are selected to be the first canonical projections. In particular, If  $q$ contains only one quantity, then $h(q) = q$. Further more, if the joint distribution of $q$ and $d$ is Gaussian, then the equality will hold. As we can see, the mutual information between the first canonical projections is a lower bound of the original mutual information criterion. And we will use this lower bound as the target function in the our sensor placement.

\subsubsection{Bayesian Optimization}

By projecting the observable and QoI onto a lower dimensional space, the computation of  mutual information becomes more reliable. However, to select the location with maximum mutual information over a continuous domain is still challenging, since mutual information is usually computed through Monte Carlo method and doesn't have a close form. A common way to solve this problem is to discretize the domain into grid then compare mutual information at each grid point. The drawback of this method is obvious. On one hand, if the grid is too fine, it will require vast computations which is not efficient. On the other hand, sparse grid may fail to capture the optimal point. In this paper, we apply Bayesian optimization~\citep{Jones1998} to facilitate the selection of sensor locations.

The basic idea of Bayesian optimization is to evaluate objective function $f(x)$ (which is mutual information in our case) at a small number of points where $f(x)$ is most likely to reach the maximum. Two concepts are need to address, one is probabilistic prior on objective function $f(x)$, the other is acquisition function $a(x)$. 

\begin{equation}
f(x) = I(q;d|x)
\end{equation}

\paragraph{Gaussian process}. Like prior on variables, we use $p(f(x))$ to denote the prior distribution of $f(x)$. After calculating the estimates of the mutual information at various scenarios, $D_{t} = \lbrace(x_{i}, f_{i})| i = 1,2, ... , t\rbrace$, the posterior distribution of $f(x)$ can be obtained through Bayes' rule
\begin{equation}
p(f(x)| D_{t}) = \dfrac{p(D_{t}| f(x))p(f(x))}{p(D_{t})}~.
\end{equation}
The most common prior is Gaussian process which is denoted as 
\begin{equation}
f(x) \sim \mathcal{GP}(m(x), k(x, x^{'}))~.  \nonumber
\end{equation}
Here, $m(x)$ is the mean function and $k(x, x^{'})$ is the covariance function. A Gaussian prior on $f(x)$ means that $f(x)$ at different points $x$ follows a multivariate Gaussian distribution whose mean is $m(x)$ and covariance matrix is specified by $k(x, x^{'})$. 
\begin{equation}
\begin{bmatrix}
f(x_{1}) \\ \vdots \\ f(x_{n})
\end{bmatrix}
\sim 
\mathcal{N}\left(
\begin{bmatrix}
m(x_{1}) \\ \vdots \\ m(x_{n})
\end{bmatrix},
\begin{bmatrix}
k(x_{1}, x_{1}) & \cdots & k(x_{1}, x_{n}) \\ \vdots & \ddots & \vdots \\k(x_{n}, x_{1}) & \cdots & k(x_{n}, x_{n})
\end{bmatrix}
\right)  \nonumber
\end{equation}
Covariance function $k(x, x^{'})$ defines the correlation between two different points $x$ and $x^{'}$. A commonly used covariance function is squared exponential function,
\begin{equation}
k(x, x^{'}) = exp(-(x - x^{'})^{T}\Lambda^{-1}(x - x^{'}))
\end{equation}
where $\Lambda$ is a diagonal matrix of which each entry $\lambda$ on the diagonal specifies the correlation length of two points in the corresponding dimension. A large $\lambda$ means $f(x)$ is smooth in that dimension. Usually, $\lambda$ is obtained through maximum likelihood estimation.

Now let $F_{t}$ denote observed outputs over $X_{t} = [ x_{1},$ $ x_{2}, ... , x_{t} ]$, $F_{n}$ denote outputs over any $X_{n} = [ x_{1}^{'}, x_{2}^{'}, ...,$  $x_{n}^{'} ]$, according to the Gaussian prior, we have
\begin{equation}
\begin{bmatrix}
F_{t} \\ F_{n}
\end{bmatrix}
\sim 
\mathcal{N}\left(
0,
\begin{bmatrix}
K(X_{t}, X_{t}) & K(X_{t}, X_{n}) \\K(X_{n}^{'}, X_{t}) & K(X_{n}, X_{n})
\end{bmatrix}
\right)~.  \nonumber
\end{equation}
$K(X_{t}, X_{n})$ is a covariance matrix of which the element at $(i, j)$ is specified by $k(x_{i}, x_{j}^{'})$.  From this joint Gaussian distribution, we can obtain the predictive distribution of $F_{n}$~\citep{Rasmussen06}
\begin{equation}
p(F_{n}| X_{n},X_{t},F_{t}) \sim \mathcal{N}(\mu(X_{n}), \Sigma(X_{n}))
\end{equation}
where
\begin{align}
\mu(X_{n}) = & m(X_{n}) + K(X_{n}, X_{t})K(X_{t}, X_{t})^{-1}(F_{t} - m(X_{t})) \label{eq:BOmu}\\
\Sigma(X_{n}) = & K(X_{n}, X_{n})- K(X_{n}, X_{t})K(X_{t}, X_{t})^{-1}K(X_{t}, X_{n})~. \label{eq:BOsigma}
\end{align}
In this way, we can get estimation of mean $\mu(f(x))$ and variance $\sigma(f(x))$ for each $x$. Thus a surface of $f(x)$ with uncertainty is generated. Each time after  evaluating $f(x)$ at some $x$, the surface will be updated according to Eq.~\eqref{eq:BOmu} and Eq.~\eqref{eq:BOsigma}.

\paragraph{Acquisition function}. In the last section, we have discussed Gaussian process prior over objective function $f(x)$ and how to update distribution of $f(x)$ given evaluation at $x$. But where to collect real data $D$ remains unsolved, since we want to place the sensor where the maximum $f(x)$ is achieved.  Here, acquisition function $a(x)$ will guide the searching process. 

Acquisition function takes the mean $\mu(f(x))$ and the variance $\sigma(f(x))$ as two arguments, and has a property that large value of  $a(x)$ is associated with  potentially large value of $f(x)$. Thus we only need to evaluate $f(x)$ at the point where $a(x)$ reaches maximum.  And $a(x)$ is usually much easier to evaluate. 

One of the most popular acquisition function is expected improvement~\citep{Brochu:2009}. Let $f(x^{+})$ denote the current maximum value of objective function. The improvement of the new evaluation will be
\begin{equation}
\delta(x) = \max \lbrace 0, f(x_{t+1}) - f(x^{+}) \rbrace~.
\end{equation}
Since $f(x_{t+1})$ follows normal distribution with mean 
$\mu(f(x))$ and variance $\sigma(f(x))$ which can be obtained from Gaussian process, the expected improvement will be
\begin{align}
EI(x) &= \int_{\delta=0}^{\delta=\infty} \delta \dfrac{1}{\sqrt{2\pi\sigma(x)}}exp\left(-\dfrac{(\mu(x) - (\delta + f(x^{+})))^{2}}{2\sigma^{2}(x)} \right)\mathrm{d}\delta \nonumber\\
&= \sigma(x)\left[ \dfrac{\mu(x) - f(x^{+})}{\sigma(x)}\Phi\left( \dfrac{\mu(x) - f(x^{+})}{\sigma(x)} \right) + \phi \left(\dfrac{\mu(x) - f(x^{+})}{\sigma(x)} \right)\right]~.
\end{align}
Maximizing $EI(x)$ will give the next sampling point/scenario to evaluate $f(x)$. The acquisition function achieves a trade-off between large values of $f(x)$ and large uncertainty in $f(x)$.

\subsubsection{Placement of Multiple Sensors}

Bayesian optimization is powerful in looking for one single sensor location. However, in most cases, it is necessary to place multiple sensors so as to get enough information about QoI. Selecting multiple sensor locations is more complicated. When using sparse grids, the number of combinations of sensor locations where mutual information needs to be evaluated becomes prohibited computationally. Here, we adopt a greedy approach that leverages the previously introduced lower bound. Suppose $N$ sensors need to be placed. When we place the first sensor,  the lower bound of $I(d_{1}(x); q)$ is maximized through Bayesian optimization and the point $x$ with maximum mutual information is selected as the first sensor's location. For the second sensor,  the lower bound of $I(d_{1}, d_{2}(x); q)$ is maximized. This time $d_{1}$ is associated with the first sensor location which is fixed.  The other sensors' locations are selected in the same way, maximizing the lower bound of $I(d_{1}, d_{2}, ... ,d_{i}(x); q)$, until $i = N$.  Here, we use $I(h_{i}(q), g_{i}(d_{i}^{*}(x)))$ to denote the lower bound and $d_{i}^{*}(x) = [d_{1},d_{2}...,d_{i}(x)]$. The whole process of sensor placement is shown in Algorithm.~\ref{BOalgorithm}.
\begin{algorithm}[H]
\caption{Sensor placement with Bayesian optimization}
\label{BOalgorithm}
\begin{algorithmic}[1]
\For{$i = 1, 2, ... , N$ }
\State Evaluate $I(h_{i}(q), g_{i}(d_{i}^{*}(x)))$  at initial points $X_{init} = \lbrace x_{k}| k = 1,2, ... , K \rbrace$.
\State Update Gaussian process on $I(h_{i}(q), g_{i}(d_{i}^{*}(x)))$ with $E_{init} = \lbrace (I_{k}, x_{k})| k = 1,2, ... , K \rbrace$.
\For{$j = 1, 2, ... , M$}
\State Select next point $x_{k+1}$ by maximizing acquisition function $a(x)$ associated with Gaussian process on $I(h_{i}(q), g_{i}(d_{i}^{*}(x)))$
\State Evaluate $I(h_{i}(q), g_{i}(d_{i}^{*}(x)))$ at $x_{k+1}$.
\State Update Gaussian process on  $I(h_{i}(q), g_{i}(d_{i}^{*}(x)))$ with $(I_{k+1}, x_{k+1})$.
\EndFor
\State Select the point $x$ with maximum value among $M + K$ evaluations as the $i$th sensor's location.
\EndFor
\end{algorithmic}
\end{algorithm}
Note that sensor locations are selected greedily by running simulations without any real data. Once observation data are collected, sensors can either stay still or move around. If mobile sensors are used, Algorithm.~\ref{BOalgorithm} need to be conducted after each Bayesian update.

\section{Simulation}
\label{Simulation}

In this section, a chemical release accident is simulated. The accident occurs in a pipeline and the plume advects and diffuses over the affected area. Sensors are placed to locate the release source. In this simulation, sensor locations are fixed during the experimental period and observation data are collected with a fixed time interval. At the same time, Bayesian inference is performed, where the estimation of release location is updated at each time point. The performance of the proposed method is assessed by looking at the uncertainty of the posterior distribution, which is quantified by entropy. We also compare sensor locations selected by Bayesian optimization and those selected on a prefixed grid.

\subsection{Dispersion Model}

In this simulation, we adopt a 2D Gaussian puff based model from~\citep{Reddy2007, TerejanuP_IF_2007} where the model is used for studying data assimilation in atmospheric dispersion. The chemical material sequentially released at the source is represented by a series of circular puffs. The advection and diffusion of the plume is decided by meteorological conditions. The concentration of each puff has a Gaussian-shape distribution and the concentration at each spacial point is the summation of contributions from all the puffs. For simplicity, deposition and puff splitting is ignored. 

Each puff is characterized by three state variables: the center of the puff $\left( X, Y\right) $, the radius $r$ and the mass $Q$. The advection is decided by the wind at the puff center and the radius of the puff is computed based on Pasuill parameterization~\citep{Reddy2007, TerejanuP_IF_2007}. The dynamics of the $k$th puff at time $t_i$ is as follows:
\begin{align}
  X_k(t_{i})  &=X_k(t_{i-1}) + W_{spd}\cos (W_{dir})\Delta T \\
  Y_k(t_{i})  &=Y_k(t_{i-1}) + W_{spd}\sin (W_{dir})\Delta T \\
  S_k(t_{i})  &=S_k(t_{i-1}) + W_{spd}\Delta T \\
  r_k(t_{i}) &= p_y S_k(t_{i})^{q_y} \\
  Q_k(t_{i}) &= Q_k(t_{i-1}) \label{process_end}
\end{align}
where $W_{spd}$ and $W_{dir}$ denote wind speed and wind direction respectively. $\Delta T$ is the time interval from $t_{i-1}$ to $t_{i}$ and $S_k(t_{i})$ is the distance the $k$th puff has transported in $\Delta T$. $p_y$ and $q_y$ are Karlsruhe-J\"{u}lich diffusion coefficients ~\citep{Reddy2007, TerejanuP_IF_2007} which specifies meteorological conditions.

For the Gaussian shaped concentration, the mean is the puff center $\left( X, Y\right)$ and the standard deviation is the radius $r$. At each spatial point $(x,y)$, the concentration is computed by summing up contributions of all the puffs
\begin{align}
&c_{(x,y)}(t_i)= \sum_{k=1}^K \frac{Q_k(t_{i})}{2\pi(r_k(t_{i}))^2}\exp{\left(-\frac{(X_k(t_{i})-x)^2+(Y_k(t_{i})-y)^2}{2(r_k(t_{i}))^2}\right)}
\end{align}
where $K$ is the number of puffs released by the time point $t_{i}$. Fig.~\ref{dispersion_onePuff} shows the advection and diffusion of one Gaussian puff. Note that the color bar of each puff is different. In fact, due to the diffusing process, the average concentration is decreasing.
 \begin{figure}[htbp]
    \centering
   \includegraphics[width=1\linewidth]{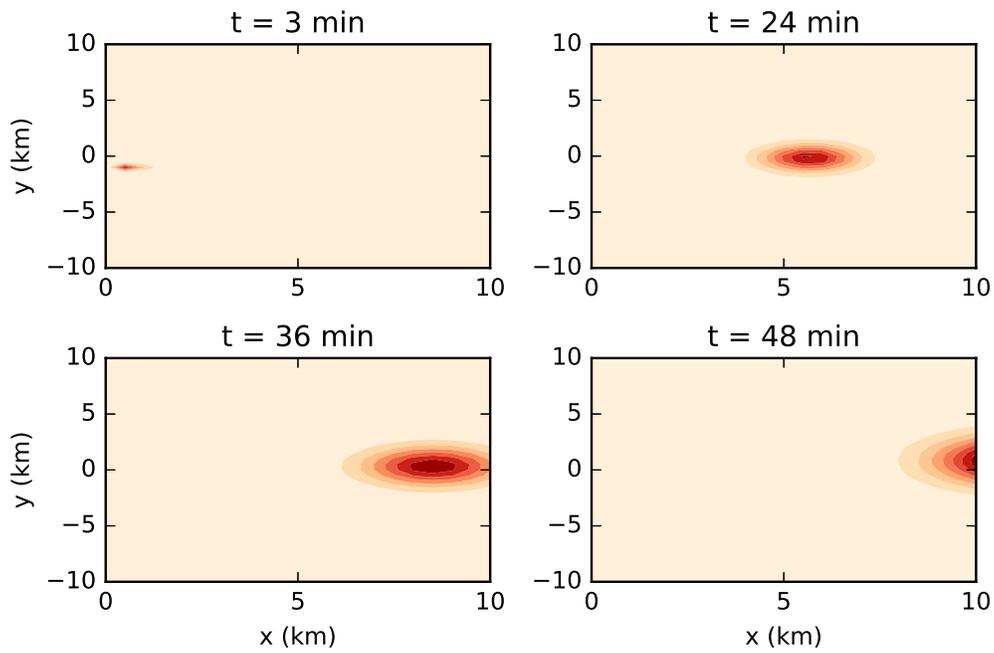}
   \caption{\textbf{The dynamics of one puff.} This figure shows the concentration surface of one puff at different time. The release occurs at t = 0 and the release location is (0, -1163.5 km). The wind direction is 0.17 m/s and the wind speed is 4 m/s. $p_y$ and $q_y$ are 0.466 and 0.866 respectively. Note the color bar is different for each plot.}
   \label{dispersion_onePuff}
\end{figure}

The following measurement model is used to relate model predictions with measurement data for each of the $N$ sensors. 
\begin{eqnarray}
ln(d_j) &=& ln(c_j) + \epsilon_{meas_{j}}  \quad j = 1,2,...,N~.  \nonumber \label{dispersion-obs} \\
\epsilon_{meas_{j}} &\sim& \mathcal{N}(-0.005, 0.1^2) \nonumber
\end{eqnarray}

\subsection{Simulation Settings}
\label{subsec:Simulation}

The affected area is a domain of $10\times20~km^2$. The pipeline is from $(0,-3 km)$ to $(0,3 km)$. The release accident can occur at any location along the pipeline. For simplicity, we assume that there is only one release source. For comparison purpose, we also grid the domain and select sensor locations on the grid points. The domain is shown in Fig.~\ref{domain}.
\begin{figure}[htbp]
  \centering
  \includegraphics[width=1\linewidth]{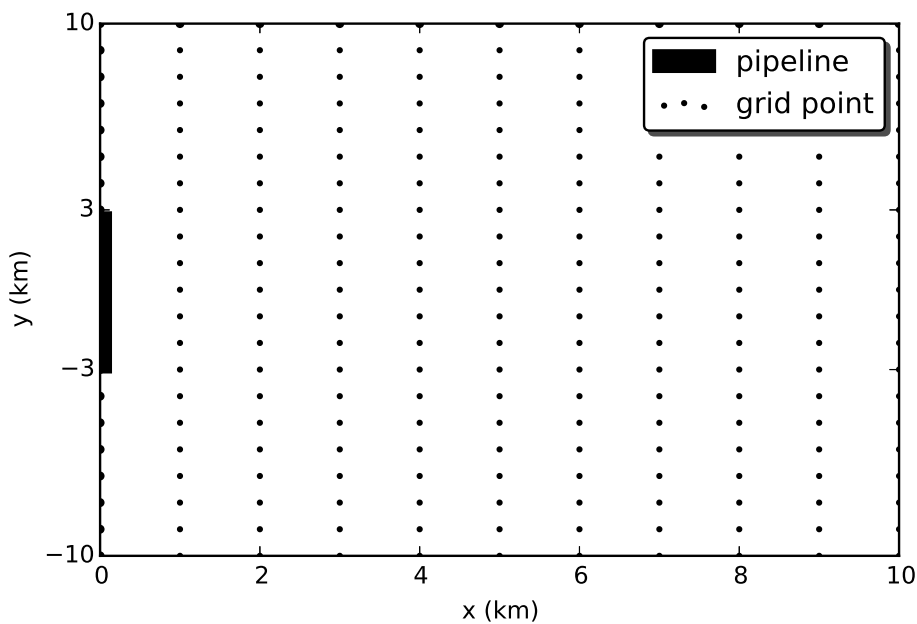}
  \caption{\textbf{The grid and release location.} The domain is grided by $11\times×21$. The pipeline is on the y axis, ranging from -3 km to 3 km.}
  \label{domain}
\end{figure}

The total simulation time is $30~min$ and the sampling interval is $1~min$.  Data collection starts from $t_1$, since in the beginning there are no chemicals in the sensing area. Each time after the data is collected, Bayesian inference is performed. For the first $10~min$, a series of Gaussian puffs are released every $1~min$ at the source from $t_0$. The wind is from west with $10^\circ$ standard deviation and the wind speed is $4~m/sec$. In this simulation, there are two unknown model parameters, release location and wind direction which both need to be inferred given measurement data. Although release location is the main concern, here we will include both parameters in $q$. 

\subsection{Selection of Sensor Locations}
\label{subsec:Find_opt}

Suppose $N$ sensors need to be placed to collect data, we use $D$ to denote the whole vector of observation data, then $D = \lbrace d_{1}, d_{2},..., d_{N} \rbrace$.  As discussed before, optimal sensor locations can be selected by computing $I(D;q|x)$. However, the optimal sensor locations are very likely to change as the plume advects and diffuses over the district. Assume sensor data is collected at a series of discrete time $t_{1}, t_{2},..., t_{M}$, then for each sensor $i, i=1,2,...,N $, $d_{i} = [d_i^{t_1}, d_i^{t_2},..., d_i^{t_M}]$. Therefore, $x$ which will provide the overall most information about $q$ should be selected in the following way:
\begin{align}\label{MI over time}
x^* = \arg\max_{x \in x}
    I(d_{1},d_{2},...,d_{N};q|x)~.
\end{align}
There are several issues in computing Eq.~\eqref{MI over time}. First, as we know, each $d_{i}$ is an $M$ dimensional vector, which means $I(d_{1},d_{2},...,d_{N};q|x)$ needs to be computed in over $NM$ dimensions.  Second, since the domain is continuous and we don't have an analytical form of mutual information, it is difficult to select best sensor locations, not mentioning multiple sensors need to be placed at the same time.  As discussed in the last section, these problems can be all addressed with our proposed method.

In this simulation, three sensors are placed. Sensor readings are collected during the simulation period to infer $q$. Since sensors are placed before any data is collected, which means the optimal sensor placement should provide most information about the release location in average sense for all possible initial conditions (release location and wind direction). The initial $1000$ ensemble members are generated at random for the release location and the wind direction from a uniform distribution over the length of the pipeline and $\mathcal{N}(0, 10^\circ)$ respectively. 

Three sensor locations are selected according to Algorithm~\ref{BOalgorithm}.  The three locations selected  are (4.8 km, -2.8 km), (2.8 km, 2.6 km) and (3.4 km, 4.1 km), which are shown in Fig.~\ref{sensor_placement}, represented by green diamonds. We also computed mutual information over the grid, the heat map in Fig.~\ref{sensor_placement} shows the value of mutual information at each grid point. Three locations where mutual information has largest value are selected, represented by blue diamonds. We can see that sensor locations selected by both methods are spatially close. 

 \begin{figure}[htbp]
    \centering
   \includegraphics[width=1\linewidth]{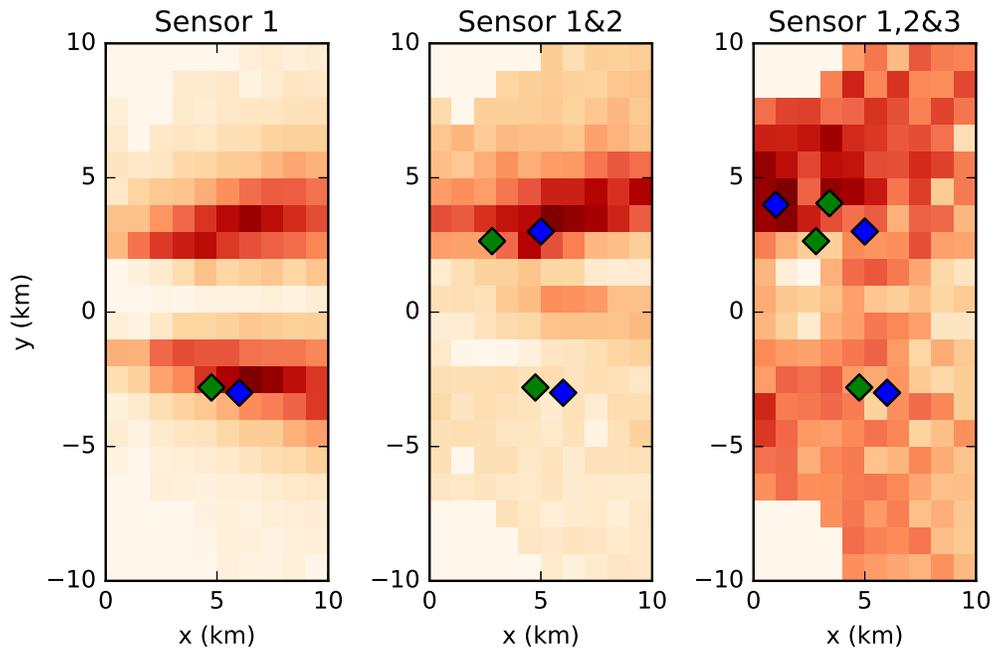}
   \caption{\textbf{Mutual information surface and sensor locations.} Three sensor locations are selected in a greedy way. The heat map show mutual information $I(d_{1},t;q)$, $I(d_{1},d_{2},t;q)$ and $I(d_{1},d_{2},d_{3},t;q)$ at each grid point. Blue diamonds represent sensor locations selected over grid points, while green ones are selected by Bayesian optimization.  Note the color bar for each plot is different.}
   \label{sensor_placement}
\end{figure}

\subsection{Result Analysis}
\label{subsec:Result_com}

\subsubsection{Estimation of Release Location}

After sensor data is collected and Bayesian inference is performed at each time point to estimate the release location. In this simulation, EnKF is used for Bayesian inference. How to use EnKF in our problem is detailed in Appendix. Here we simulate a release accident and see how the posterior distribution of the release location changes. The real release location is (0, -1291.7 m), where the wind direction is -0.026 rad, both of which are randomly selected from their distributions. Observation data is collected and the distribution of the release location is updated at each time point from $t_{1}$ to $t_{31}$. Fig.~\ref{figs_EnKF_BestR} shows the initial prior distribution and  posterior distributions after update at $t_{11}$, $t_{21}$ and $t_{31}$. It is obvious that the real release location is captured by the posterior distribution and the uncertainty is decreasing as more data are collected. One may note that the change in uncertainty from $t_{0}$ to $t_{11}$ is much bigger than that from $t_{21}$ to $t_{31}$, which means most uncertainty is reduced in the first several updates. Due to the measurement noise and computational approximations, the uncertainty of release location can never be eliminated but might converge to a small level.
\begin{figure}[htbp]
    \centering
    \begin{subfigure}[t]{0.8\textwidth}
        \centering
        \includegraphics[width=1\textwidth]{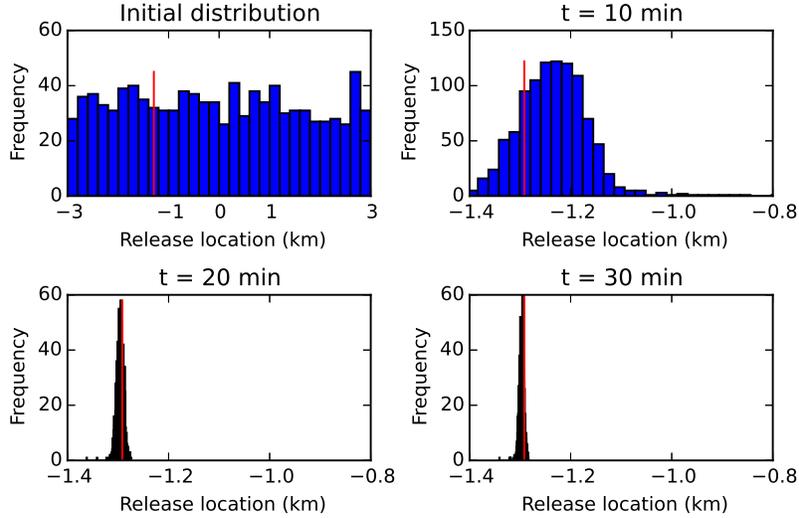}
        \caption{Prior $\&$ posterior distribution of release location}
    \end{subfigure}%
    \hspace{0.1in}  
    \begin{subfigure}[t]{0.8\textwidth}
        \centering
        \includegraphics[width=1\textwidth]{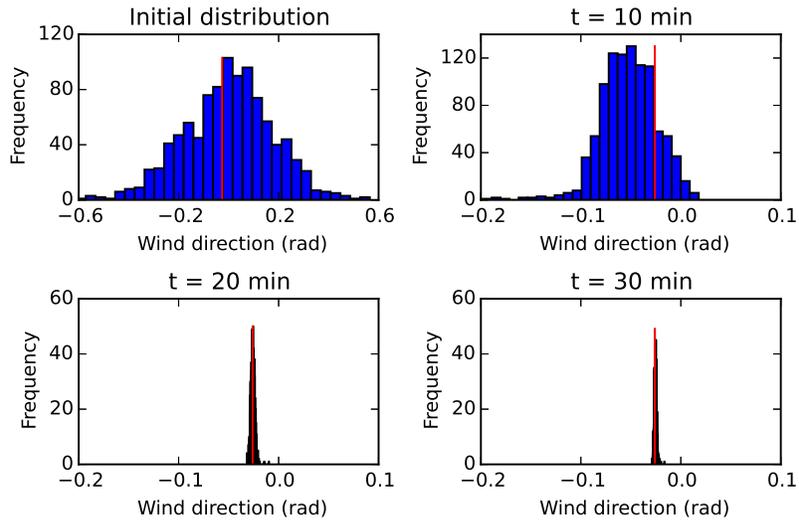}
        \caption{Prior $\&$ posterior distribution of wind direction}
     \end{subfigure}%
    \caption{\textbf{Distributions of parameters} The histograms in Figure (a) are prior distribution and posterior distribution of release location at each time step. Figure (b) shows corresponding distribution of wind direction. The real release location is (0,-1291.7 km) and the real wind direction is -0.026 rad. Both of them are represented by red vertical lines.}
     \label{figs_EnKF_BestR}
\end{figure}

\subsubsection{Comparison with Other Placements}

Here, we compare the selected sensor locations with 20 random placements under 50 different initial conditions. Each initial condition is a combination of release location and wind direction both of which are randomly sampled from their distributions. The performance is measured by conditional entropy. For each placement, the entropy of $q$ is given by $H(q|\xi=\xi_{i}, x=x_{j})$. Here, $x$ represent the sensor location and the subscript $j$ denotes a specific placement.  $\xi$ is initial condition and $\xi_{i}$ implies a particular initial condition. To compare average performance under different initial conditions, conditional entropy $H(q| \xi, x = x_{j})$ is computed by Eq.~\eqref{eq:conditionalEN}.
\begin{align}
H(q| \xi, x = x_{j}) = \sum_{\xi_{i}}p(\xi=\xi_{i})H(q|\xi=\xi_{i}, x=x_{j}) 
\label{eq:conditionalEN}
\end{align}

Fig.~\ref{figs_EntValue_Com} shows the average performance of each placement over 50 different initial conditions. We can see that MI via grids and MI via Bayesian optimization have similar performance, and they both outperforms most of random placements. They not only lead to smaller uncertainty in posterior distribution but also show a faster reduction in uncertainty.
\begin{figure}[htbp]
    \centering
    \begin{subfigure}[t]{0.8\textwidth}
        \centering
        \includegraphics[width=1\textwidth]{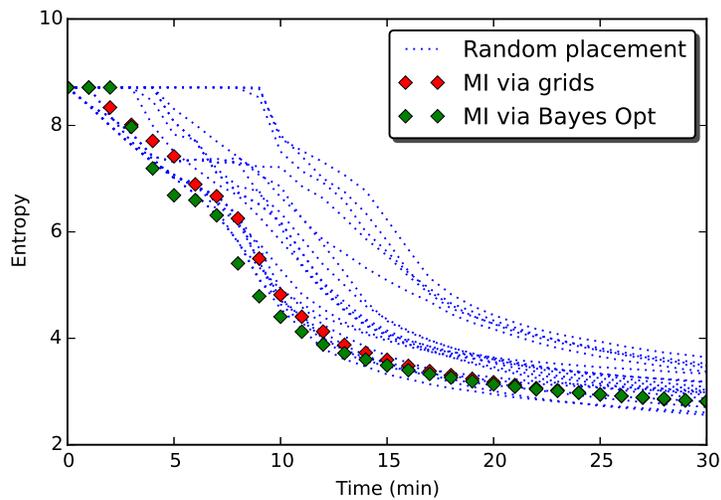}
        \caption{Entropy of release location}
    \end{subfigure}%
    \hspace{0.1in}  
    \begin{subfigure}[t]{0.8\textwidth}
        \centering
        \includegraphics[width=1\textwidth]{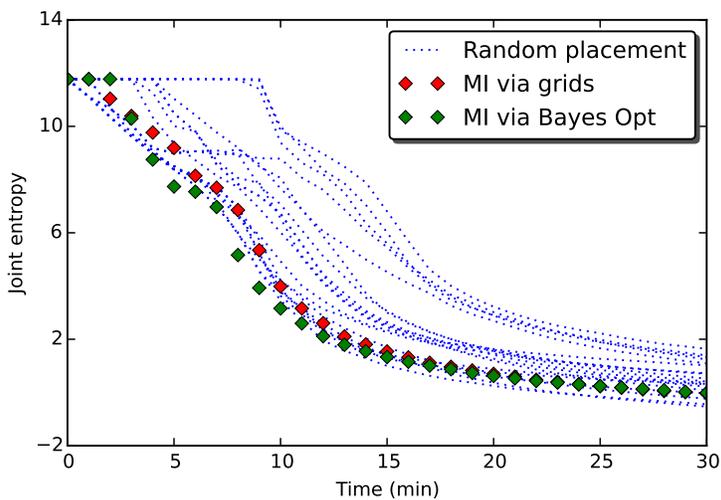}
        \caption{Joint entropy}
     \end{subfigure}%
    \caption{\textbf{Results comparison}. These two figures compare the performance of proposed strategy with 20 random sensor placements. Here, entropy is used to measure the uncertainty. Figure (a) shows entropy of release location after each update and Figure(b) shows joint entropy of release location and wind direction}
    \label{figs_EntValue_Com}
\end{figure}

\section{Conclusions}
\label{Conclusion}

In this paper, we address the sensor placement problem where certain QoI need to be inferred from observation data. The QoI often include model parameters and/or state variables but can also be other quantities. Inferring QoI usually involves solving inverse problem which is formulated in Bayesian framework in this paper. On the other hand, since data collected at different locations are likely to provide different amount of information, sensors should be strategically placed.

Mutual information is one of the most common used criterion to guide the sensor placement. It naturally quantifies the dependence between two variables and has been widely discussed in the literature. However, computing mutual information is challenging and the estimation in high dimension is always unreliable. Thus we propose a novel approach, which compute the lower bound of mutual information in only two dimension. The approach is based on data inequality processing and canonical correlation analysis. It projects observation and QoI into two dimension where the projections have largest correlation. This lower bound of mutual information is shown to be effective as a metric to select sensor locations. In addition, we apply Bayesian optimization to deal with continuous domain. We place Gaussian prior on the metric, which generates a mutual information surface. Then evaluations are made according to acquisition function. In this way, the number of evaluations is greatly reduced. 

A chemical dispersion accident is simulated and it shows that the proposed approach outperforms random sensor placements by offering an obviously faster reduction in uncertain of QoI. Also the sensor locations selected by Bayesian optimization are close to those by discretizing the domain into a fine grid, but with a considerably less number of evaluations.  The proposed approach is promising to address a vast range of sensor placement problems.

\section*{Acknowledgments}
This material is based upon work supported by the National Science Foundation under Grant No. 1504728, Shannxi Province National Science Fund 2013KW22-02 and
Xi'an Science and Technology Project Fund CX1252(7). We would like to thank Dr. Brian J. Williams of Los Alamos National Laboratory for his valuable suggestions and comments.

\appendix
\section{EnKF}

In this simulation, the inverse problem of inferring unknown parameters is solved by EnKF.  We will first give a brief introduction of EnKF, then show how to apply it to solve our problem.

EnKF uses ensemble members to describe probability distribution. Each ensemble member is updated through similar procedures as in Kalman filter. Let $u$ denote the state of the system. Suppose the observation model is
\begin{equation}
     d= Hu + \epsilon~.
\end{equation}
EnKF updates each ensemble member of $u$ as follows:
\begin{align}
 &u_{j}(t_i) = u_{j}(t_{i-1}) + \Sigma_{e}(t_{i-1})H^{T} \times \nonumber \label{update}\\
& \quad \quad \quad \quad [H\Sigma_{e}(t_{i-1})H^{T} + R_{e}]^{-1}[d_{j}(t_{i}) - Hu_{j}(t_{i-1})] \\
 &\Sigma_{e}(t_{i-1}) = \overline{[u(t_{i-1}) - \overline{u(t_{i-1})}][u(t_{i-1}) - \overline{u(t_{i-1})}]^{T}}\\
& \Sigma_{e}(t_{i}) = \overline{[u(t_{i}) - \overline{u(t_{i})}][u(t_{i}) - \overline{u(t_{i})}]^{T}}
\end{align}
where the perturbed measurements $d_{j}(t_{i}) = d(t_{i}) + \epsilon_{j}$,  $j=1\ldots Q$. Here $Q$ is the size of the ensemble and  $\epsilon_{j}$ is a sample of  $\epsilon$. $R_{e} = \overline{\epsilon\epsilon^{T}}$ is the covariance matrix of the noise samples.

In our problem, the states of the system are the concentrations at sensor locations. There are two parameters, release location and wind direction, which also need to be updated along with the concentrations $u = [u_{1},u_{2},...,u_{N}]$, an augmented state is required. Let $\theta$ denote parameters. Since observation noise is multiplicative, the augmented state is
\begin{equation}
     u^{*} = [\theta \quad ln{u}]^{T} \nonumber
\end{equation}
Another issue is that the mean of the measurement noise is nonzero as it follows lognormal distribution. Here we replace the original noise with a bias term $\mu$ plus a zero-mean noise $\epsilon^{*} \sim \mathcal{N}(0, \sigma^{2})$. Then the observation model becomes
\begin{equation}
     d^{*}= Hu^{*} + \mu+ \epsilon^{*}
\end{equation}
where $d^{*} = ln(d^{T})$ and $H = [0_{2, N} \quad  I_{N}] $. $I_{N}$ is  an identity matrix, $0_{2, N}$ is a zero matrix.  Last, the residue $d_{j}(t_{i}) - Hu_{j}(t_{i-1})$  in Eq.~\eqref{update} needs to be replaced by $d_{j}(t_{i}) - Hu_{j}(t_{i-1}) - \mu$. Given the joint posterior distribution of state and parameter,  the posterior distribution of the $\theta$ can be obtained by marginalization. 

\clearpage

\section{References}
\label{References}

\bibliographystyle{elsarticle-harv}
\bibliography{main}

\begin{thebibliography}{28}
\expandafter\ifx\csname natexlab\endcsname\relax\def\natexlab#1{#1}\fi
\expandafter\ifx\csname url\endcsname\relax
  \def\url#1{\texttt{#1}}\fi
\expandafter\ifx\csname urlprefix\endcsname\relax\def\urlprefix{URL }\fi

\bibitem[{Bach and Jordan(2003)}]{Bach:2003}
Bach, F.~R., Jordan, M.~I., March 2003. Kernel independent component analysis.
  J. Mach. Learn. Res. 3, 1-- 48.

\bibitem[{Beck and Au(2002)}]{Beck:2002}
Beck, J., Au, S., 2002. {Bayesian Updating of Structural Models and Reliability
  using Markov Chain Monte Carlo Simulation}. J. Eng. Mech. 128~(4), 380--391.

\bibitem[{Brochu et~al.(2009)Brochu, Cora, and {de Freitas}}]{Brochu:2009}
Brochu, E., Cora, V.~M., {de Freitas}, N., 2009. A tutorial on {B}ayesian
  optimization of expensive cost functions, with application to active user
  modeling and hierarchical reinforcement learning. Tech. Rep. UBC TR-2009-023
  and arXiv:1012.2599, University of British Columbia, Department of Computer
  Science.
\newline\urlprefix\url{http://arxiv.org/abs/1012.2599}

\bibitem[{Cheng et~al.(2013)Cheng, Wang, Morelande, and
  Moran}]{cheng2013information}
Cheng, Y., Wang, X., Morelande, M., Moran, B., 2013. Information geometry of
  target tracking sensor networks. Information Fusion 14~(3), 311--326.

\bibitem[{Ching and Chen(2007)}]{Ching:2007}
Ching, J., Chen, Y., 2007. {Transitional Markov Chain Monte Carlo Method for
  Bayesian Model Updating, Model Class Selection, and Model Averaging}. J. Eng.
  Mech. 133~(7), 816--832.

\bibitem[{Cover and Thomas(2006)}]{Cover:2006}
Cover, T.~M., Thomas, J.~A., 2006. Elements of Information Theory (Wiley Series
  in Telecommunications and Signal Processing). Wiley-Interscience.

\bibitem[{Dhillon and Chakrabarty(2003)}]{dhillon2003}
Dhillon, S.~S., Chakrabarty, K., 2003. Sensor placement for effective coverage
  and surveillance in distributed sensor networks. In: Wireless Communications
  and Networking, 2003. WCNC 2003. 2003 IEEE. Vol.~3. IEEE, pp. 1609--1614.

\bibitem[{Ertin et~al.(2003)Ertin, Fisher, and Potter}]{Ertin2003}
Ertin, E., Fisher, J., Potter, L., 2003. Maximum mutual information principle
  for dynamic sensor query problems. In: Information processing in sensor
  networks. Springer, pp. 405--416.

\bibitem[{Evensen(2009)}]{Evensen2009}
Evensen, G., aug 2009. {Data Assimilation}. The Ensemble Kalman Filter.
  Springer.

\bibitem[{Haario et~al.(2006)Haario, Laine, Mira, and Saksman}]{Haario:2006}
Haario, H., Laine, M., Mira, A., Saksman, E., 2006. {DRAM: Efficient adaptive
  MCMC}. Statistics and Computing 16~(4), 339--354.

\bibitem[{Hutchinson et~al.(2017)Hutchinson, Oh, and Chen}]{hutchinson2017}
Hutchinson, M., Oh, H., Chen, W.-H., 2017. A review of source term estimation
  methods for atmospheric dispersion events using static or mobile sensors.
  Information Fusion 36, 130--148.

\bibitem[{Jones et~al.(1998)Jones, Schonlau, and Welch}]{Jones1998}
Jones, D.~R., Schonlau, M., Welch, W.~J., 1998. Efficient global optimization
  of expensive black-box functions. Journal of Global Optimization 13,
  455–492.

\bibitem[{Khaleghi et~al.(2013)Khaleghi, Khamis, Karray, and
  Razavi}]{khaleghi2013multisensor}
Khaleghi, B., Khamis, A., Karray, F.~O., Razavi, S.~N., 2013. Multisensor data
  fusion: A review of the state-of-the-art. Information Fusion 14~(1), 28--44.

\bibitem[{Khan et~al.(2007)Khan, Bandyopadhyay, Ganguly, Saigal, Erickson,
  Protopopescu, and Ostrouchov}]{Khan:2007up}
Khan, S., Bandyopadhyay, S., Ganguly, A.~R., Saigal, S., Erickson, III, D.~J.,
  Protopopescu, V., Ostrouchov, G., 2007. {Relative performance of mutual
  information estimation methods for quantifying the dependence among short and
  noisy data}. Physical Review E 76~(2), 026209.

\bibitem[{Kozachenko and Leonenko(1987)}]{Kozachenko:1987ts}
Kozachenko, L.~F., Leonenko, N.~N., 1987. {L. F. Kozachenko, N. N. Leonenko,
  ``Sample Estimate of the Entropy of a Random Vector'', Probl. Peredachi Inf.,
  23:2 (1987), 9--16}. Problemy Peredachi Informatsii.

\bibitem[{Kraskov et~al.(2004)Kraskov, St{\"o}gbauer, and
  Grassberger}]{Kraskov:2004gr}
Kraskov, A., St{\"o}gbauer, H., Grassberger, P., jun 2004. {Estimating mutual
  information}. Physical Review E 69~(6), 066138.

\bibitem[{Krause et~al.(2008)Krause, Singh, and Guestrin}]{Krause2008}
Krause, A., Singh, A., Guestrin, C., Jun. 2008. Near-optimal sensor placements
  in gaussian processes: Theory, efficient algorithms and empirical studies. J.
  Mach. Learn. Res. 9, 235--284.

\bibitem[{Lindley(1956)}]{Lindley1956}
Lindley, D., 1956. On a measure of the information provided by an experiment.
  Ann. Math. Statist. 27(4), 986--1005.

\bibitem[{Madankan et~al.(2014)Madankan, Singla, and Singh}]{MadankanSS14}
Madankan, R., Singla, P., Singh, T., 2014. Optimal information collection for
  source parameter estimation of atmospheric release phenomenon. In: American
  Control Conference, {ACC} 2014, Portland, OR, USA, June 4-6, 2014. pp.
  604--609.

\bibitem[{Mira(2001)}]{Mira:2001}
Mira, A., 2001. {On Metropolis-Hastings algorithms with delayed rejection}.
  Metron 59, 3--4.

\bibitem[{Oudjane and Musso(2000)}]{oudjane:2000}
Oudjane, N., Musso, C., 2000. Progressive correction for regularized particle
  filters. In: Information Fusion, 2000. FUSION 2000. Proceedings of the Third
  International Conference on. Vol.~2. IEEE, pp. THB2--10.

\bibitem[{Rasmussen(2006)}]{Rasmussen06}
Rasmussen, C.~E., 2006. Gaussian processes for machine learning. MIT Press.

\bibitem[{Reddy et~al.(2007)Reddy, Cheng, Singh, and Scott}]{Reddy2007}
Reddy, K., Cheng, Y., Singh, T., Scott, P., 2007. Data assimilation in variable
  dimension dispersion models using particle filters. In: Information Fusion,
  2007 10th International Conference on. IEEE, pp. 1--8.

\bibitem[{Terejanu et~al.(2007)Terejanu, Singh, and Scott}]{TerejanuP_IF_2007}
Terejanu, G., Singh, T., Scott, P.~D., July 2007. Unscented {K}alman
  filter/smoother for a {CBRN} puff-based dispersion model. In: 11th
  International Conference on Information Fusion, Quebec City, Canada.

\bibitem[{Walters-Williams and Li(2009)}]{WaltersWilliams:2009gh}
Walters-Williams, J., Li, Y., 2009. {Estimation of Mutual Information: A
  Survey}. In: Rough Sets and Knowledge Technology. Springer Berlin Heidelberg,
  Berlin, Heidelberg, pp. 389--396.

\bibitem[{Wang et~al.(2004)Wang, Pottie, Yao, and Estrin}]{Hanbiao2004}
Wang, H., Pottie, G., Yao, K., Estrin, D., April 2004. Entropy-based sensor
  selection heuristic for target localization. In: Information Processing in
  Sensor Networks, 2004. IPSN 2004. Third International Symposium on. pp.
  36--45.

\bibitem[{Weaver et~al.(2016)Weaver, Williams, Anderson-Cook, , and
  Higdon}]{Weaver2016}
Weaver, B.~P., Williams, B.~J., Anderson-Cook, C.~M., , Higdon, D.~M., 2016.
  Computational enhancements to bayesian design of experiments using gaussian
  processes. Bayesian Analysis 11~(1), 191–213.

\bibitem[{Wu et~al.(2012)Wu, Liu, and Wu}]{Xiaopei2012}
Wu, X., Liu, M., Wu, Y., Sep. 2012. In-situ soil moisture sensing: Optimal
  sensor placement and field estimation. ACM Trans. Sen. Netw. 8~(4),
  33:1--33:30.

\end{thebibliography}





\end{document}